\documentclass[conference,letterpaper]{IEEEtran}

\usepackage[utf8]{inputenc} 
\usepackage[T1]{fontenc}
\usepackage{url}
\usepackage{ifthen}
\usepackage{cite}
\usepackage[cmex10]{amsmath} % Use the [cmex10] option to ensure complicance
\usepackage{graphicx}
\usepackage{amssymb}
\usepackage{amsthm}
\usepackage{amsfonts}
\usepackage{mathtools}
\usepackage{dsfont}
\mathtoolsset{showonlyrefs=true}
\usepackage{algorithm}
\usepackage{algorithmic}

\usepackage{color}
\usepackage{float}

\newcommand{\rem}[1]{}

\newtheorem{lemma}{Lemma}

\newtheorem{theorem}{Theorem}

\newcommand{\bre}{\begin{equation}}
\newcommand{\ere}{\end{equation}}

\newcommand{\ee}\]
\newcommand{\bra}{\begin{eqnarray}}
\newcommand{\era}{\end{eqnarray}}
\newcommand{\bfg}{\begin{figure}[hbtp]}
\newcommand{\efg}{\end{figure}}
\newcommand{\bver}{\begin{verbatim}}
\newcommand{\ever}{\end{verbatim}}
\newcommand{\bit}{\begin{itemize}}
\newcommand{\eit}{\end{itemize}}
\newcommand{\ben}{\begin{enumerate}}
\newcommand{\een}{\end{enumerate}}

\newcommand{\given}{\: | \:}

\newcommand{\bc}{{\bf c}}

\newcommand{\bm}{{\bf m}}

\newcommand{\hepsilon}{\hat{\epsilon}}

\newcommand{\baa}{\begin{eqnarray*}}
\newcommand{\eaa}{\end{eqnarray*}}

\newcommand{\bs}{{\bf s}}

\newcommand{\bX}{{\bf X}}

\newcommand{\xor}{\oplus}
\newcommand{\bu}{{\bf u}}

\newcommand{\bx}{{\bf x}}
\newcommand{\by}{{\bf y}}
\newcommand{\bz}{{\bf z}}

\newcommand{\bzr}{{\bf 0}}

\newcommand{\cS}{{\cal S}}

\newcommand{\cX}{{\cal X}}
\newcommand{\cY}{{\cal Y}}

\newcommand{\cU}{{\cal U}}

\newcommand{\cC}{{\mathcal{C}}}

\newcommand{\defined}{\triangleq}

\def\defined{\: {\stackrel{\scriptscriptstyle \Delta}{=}} \: }

\newfont{\boldlarge}{msbm10 scaled 1100}

\newcommand{\comment}[1]{}

\newlength{\tmpbigbar}

\begin{document}
\title{On Polar Coding for Binary Dirty Paper}

% %%% Single author, or several authors with same affiliation:
\author{%
\IEEEauthorblockN{Barak Beilin and David Burshtein}
\IEEEauthorblockA{School of Electrical Engineering\\
                  Tel-Aviv University\\
				  Tel-Aviv 6997801, Israel\\
                  Email: barakbei@mail.tau.ac.il, burstyn@eng.tau.ac.il}
}

\maketitle

%%%%%%
%% Abstract: 
%% If your paper is eligible for the student paper award, please add
%% the comment "THIS PAPER IS ELIGIBLE FOR THE STUDENT PAPER
%% AWARD." as a first line in the abstract. 
%% For the final version of the accepted paper, please do not forget
%% to remove this comment!
%%
\begin{abstract}
The problem of communication over binary dirty paper (DP) using nested polar codes is considered. An improved scheme, focusing on low delay, short to moderate blocklength communication is proposed. Successive cancellation list (SCL) decoding with properly defined CRC is used for channel coding, and SCL encoding without CRC is used for source coding. The performance is compared to the best achievable rate of any coding scheme for binary DP using nested codes.
A well known problem with nested polar codes for binary DP is the existence of frozen channel code bits that are not frozen in the source code. These bits need to be retransmitted in a second phase of the scheme, thus reducing transmission rate. We observe that the number of these bits is typically either zero or a small number, and provide an improved analysis, compared to that presented in the literature, on the size of this set and on its scaling with respect to the blocklength when the power constraint parameter is sufficiently large or the channel crossover probability sufficiently small.
\end{abstract}

%% The paper must be self-contained. However, if you are referring to
%% a full version for checking certain proofs, please provide the
%% publically accessible location below.  If the paper is completely
%% self-contained, you can remove the following line from your
%% submission.
%\textit{A full version of this paper is accessible at:}
%\url{http://isit2019.fr/} 

\section{Introduction} \label{sec:introduction}
Consider the problem of transmission over a side information channel with non-causal side information, also known as the Gelfand-Pinsker (GP) problem.
Applications include watermarking codes, memories with defects, write once memories and transmission over broadcast channels.
In the GP problem the encoder needs to send a message $M$ reliably over some memoryless channel $W(y\given x,s)$ where $x\in\cX$ is the input to the channel, $y\in\cY$ is the output, and $s\in\cS$ is the channel state. For each transmitted symbol, $x$, the state $s$ is obtained by i.i.d. sampling of some given source random variable $S\in\cS$. The encoder observes the channel state vector, $\bs=(s_0,s_1,\ldots,s_{N-1})$, non-causally, prior to transmission. It then constructs a codeword $\bx=(x_0,x_1,\ldots,x_{N-1})$ which is a function of the message $\bm$ and the state vector $\bs$.
The decoder observes only the vector of channel outputs, $\by=(y_0,y_1,\ldots,y_{N-1})$ and constructs the decoded codeword $\hat{\bm}$ from $\by$.

The binary dirty paper (DP) problem depicted in Figure \ref{fig:gpscheme} is a side information problem
with $\cX = \cS = \cY = \{0,1\}$, $S\sim {\rm Ber}(1/2)$ (${\rm Ber}(q)$ denotes a Bernoulli $(0,1)$ random variable with probabilities $(1-q,q)$), and the channel $W(y\given x,s)$ is defined by
\bre
Y = X \xor S \xor Z 
\ere
where $\xor$ denotes XOR, and $Z\sim {\rm Ber}(p)$ for $p\in(0,1/2)$.
\begin{figure}[htbp]
	\centering
	\includegraphics[width=0.9\columnwidth]{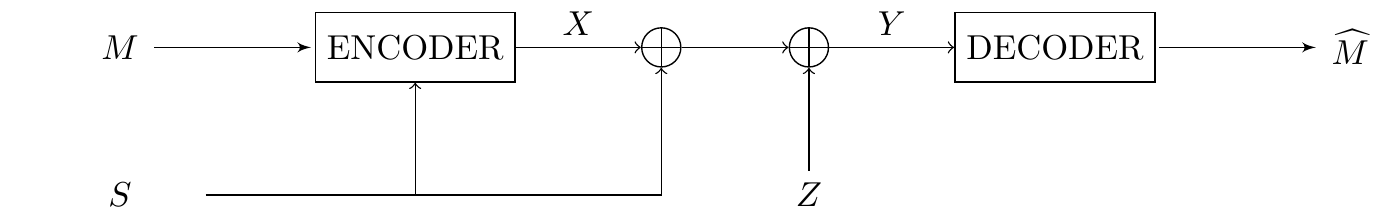} 
	\vspace{-10pt}
	\caption{The binary dirty paper problem.}
	\label{fig:gpscheme}
\end{figure}
In this problem there is a ``power constraint'' which comes in one of two possible forms. The first form is an average power constraint: Given some $D\in(0,1/2)$, on average a fraction at most $D$ of the channel input bits, $X$, are ones. That is, ${\rm E} w_H(\bX) / N \le D$ where $\bX$ is the codeword and $w_H()$ denotes Hamming weight. The second form is an individual codeword power constraint, which is stronger than the average power constraint. In this case, each codeword, $\bx$, needs to satisfy $w_H(\bx)/N \le D$.
An error event in the communication happens when either the receiver does not decode the correct transmitted message, or (under the individual codeword power constraint) when the encoder cannot obtain a codeword that satisfies the required power constraint.
For $D>p$ the capacity of the binary DP problem is given by
\bre
C_{\rm DP} = h_2(D) - h_2(p)
\ere
where $h_2(\cdot)$ is the binary entropy function (base $2$).
Following \cite{korada2010polar}, we will assume in this paper that $D>p$ since for the other case, $D<p$, the capacity as a function of $D$ can be achieved by time sharing with the point with $D=0$ and $R=0$.

In \cite[Section VIII]{korada2010polar} Arikan's polar codes \cite{arikan2009channel} were extended to the problem of binary DP using a polar nested codes structure.
Given some pair $(p,D)$ such that $0<p<D<1/2$ and blocklength $N$, one constructs two polar codes with blocklength $N$. The first is a standard binary polar code, $\cC_c$, with frozen set $F_c$,
designed for reliable communication over the binary symmetric channel, BSC($p$). The second code, $\cC_s$, with frozen set $F_s$, is a binary polar code for lossy source coding (quantization) \cite{korada2010polar} of the source $S$ with distortion $D$. The design of $\cC_s$ is very similar to the design of a polar channel code for a BSC($D$) channel, except for the threshold defining its frozen set.
The polar coding scheme for the binary DP problem was formulated in \cite{korada2010polar} under an average power constraint but it is also valid under an individual codeword power constraint \cite{bur_str_full_version}. To describe the method, first assume that $F_c \subseteq  F_s$ such that $\cC_c$ and $\cC_s$ are nested polar codes. As usual, denote by $\bu$ the message and frozen bits of a polar code, and by $\bc = \bu G_N$ the corresponding polar codeword where $G_N$ is the generating matrix defined in \cite[Eq. (70)]{arikan2009channel}. The encoder first sets $\bu_{F_c}=\bzr$ and $\bu_{F_s\setminus F_c}=\bm$ where $\bm$ denotes the information bits that need to be transmitted. This defines a polar source code that we denote by $\cC_s \left( F_s,\bu_{F_s}(\bm) \right)$. The encoder then observes $\bs$ and obtains a polar codeword $\bs' \in \cC_s \left( F_s,\bu_{F_s}(\bm) \right)$ that satisfies the power constraint (e.g., under the individual codeword power constraint, $w_H(\bs \xor \bs')\le D$) using a successive cancellation (SC) encoding algorithm which is a randomized version of the standard SC decoding algorithm \cite{arikan2009channel} (as explained in \cite{korada2010polar}, in practice the standard SC decoding algorithm can be used by the encoder without modification, but the proof requires the randomized version of the algorithm).
Now the encoder transmits $\bx = \bs \xor \bs'$.
The decoder receives $\by = \bx \xor \bs \xor \bz = \bs' \xor \bz$. Noting that $\bs' \in \cC_c \left( F_c,\bu_{F_c}=0 \right)$, and that $F_c$ was designed to obtain a good channel code for the BSC($p$), the decoder can decode $\hat{\bu}_{F_c^c}$ ($F_c^c$ denotes the complementary of $F_c$) from $\by$ using the SC decoding algorithm. The decoded message is obtained as $\hat{\bm} = \hat{\bu}_{F_s\setminus F_c}$.
The rate of this communication scheme is
$
R = (|F_s| - |F_c|)/N
$.
Since for $N$ large enough $|F_s|/N \rightarrow h_2(D)$ and $|F_c|/N \rightarrow h_2(p)$, we have $R\rightarrow C_{\rm DP}$. That is, the scheme approaches capacity.

When the assumption $F_c \subseteq F_s$ is violated, it was suggested \cite{korada2010polar} to use a two phase transmission scheme. The first phase is identical to the one described above.
In the second phase, we transmit $\bu_{F_c\cap F_s^c}$ (or accumulate the bits of $\bu_{F_c\cap F_s^c}$ of many first phase transmissions and send them together).
In the second phase the transmitter ignores the power constraint and transmits $\bx=\bc \xor \bs$ such that the state noise is canceled. However, if $|F_c \cap F_s^c|$ is small then the damage to the power constraint will be negligible (we can compensate by decreasing the design distortion $D$ of $\cC_s$). 
The decoder starts by decoding $\bu_{F_c \cap F_s^c}$ from the second phase transmission, and then it can apply the decoding described above for the case $F_c \subseteq  F_s$.
Another possibility is to transmit in frames using the chaining construction \cite{hassani2014universal} as used in \cite{mondelli2015achieving} for achieving Marton's region for broadcast channels.
When using a sufficiently large number of frames, the rate loss of the chaining construction becomes negligible. However, the use of $L$ frames increases latency by a factor of $L$, and hence is not suitable for low delay communications.

In this paper we propose an improved, low delay communication scheme over binary DP using nested polar codes \cite{korada2010polar}, with CRC aided SC list (SCL) decoding \cite{tal2015list} for channel coding, and SCL encoding without CRC for source coding. The performance is compared to the best achievable rate of any coding scheme for binary DP using nested codes. We observed that typically the set $F_c \cap F_s^c$, that needs to be retransmitted in the second phase of the scheme \cite{korada2010polar} is zero for $D-p$ larger than some small threshold. Our main theoretical contribution is an improved analysis compared to that presented in \cite{korada2010polar} on $|F_c\cap F_s^c|$ and on its scaling with respect to the blocklength when $D$ is sufficiently large or $p$ sufficiently small.

\section{Polar SCL coding for binary dirty paper}
\label{sec:BDPproblem}
We now discuss the design of an SCL coding scheme based on the SCL decoder \cite{tal2015list} and the nested polar coding scheme of \cite{korada2010polar}. We first observe that lists are useful for improved lossy source coding since the encoder can choose the least distorted codeword from several possibilities. For example, using an SCL encoder with $L_s=50$ lists to encode a Ber($1/2$) source at rate 0.258, with a polar code with blocklength $N=1024$, the distortion was 0.217, compared to 0.226 when lists are not used. The theoretical minimum distortion for this rate is 0.21. Repeating the experiment with $N=4096$ yielded a distortion of 0.215 with $L_s=50$ lists compared to a distortion of 0.223 without lists.
Although the use of CRC verification provides a significant reduction in the error rate for channel coding \cite{tal2015list}, it is not helpful for lossy source coding. This is due to the fact that in this problem we are only interested in the distortion between the source vector and any codeword. Hence a range of codewords which are sufficiently close to the source vector can be used, rather than a single preferred codeword as in the channel coding problem, making the CRC rule irrelevant. Hence, we use SCL encoder with $L_s$ lists and without CRC, and SCL decoder with $L_c$ lists and with CRC.

The other important design consideration relates to the proper definition of the CRC code. In \cite{tal2015list}, $r$ CRC bits are computed from $N-|F_c|-r$ information bits. The $r$ CRC bits are then appended to the $N-|F_c|-r$ information bits. For polar coding over side information channels it would be problematic to compute the CRC this way since this would impose a difficult constraint on the encoder side (how to output a valid codeword for channel coding that satisfies the required distortion bound). To solve this problem, we compute the $r$ CRC bits only from the message bits that need to be transmitted to the decoder.
Without using CRC there are $|F_s \cap F_c^c|$ message bits. When using CRC we reduce this number to $|F_s \cap F_c^c|-r$ and compute the $r$ CRC bits from these message bits. We then append the $r$ CRC bits to the $|F_s \cap F_c^c|-r$ message bits, and place them in $F_s \cap F_c^c$ (as in \cite{korada2010polar} we set zeros in $F_c$). At the decoder side, in the last decoding stage only those lists for which the CRC is satisfied are considered as in \cite{tal2015list}. 

As in \cite[Eqs. (18)-(19)]{korada2010polar}, the frozen sets of the codes $\cC_s$ and $\cC_c$ are defined by
\begin{align}
F_s &= \left\{ i \: : \: Z^{(i)}_N(D)  \ge \delta_N(D) \right\} \label{eq:Fs}\\
F_c &= \left\{ i \: : \: Z^{(i)}_N(p)  \ge \delta_N(p) \right\} \label{eq:Fc}
\end{align}
where $Z_N^{(i)}(D)$ ($Z_N^{(i)}(p)$, respectively) is the Bhattacharyya parameter of the $i$-th sub-channel after\footnote{The basis of all the logarithms in this paper is 2.} $n=\log N$ polarization steps of a BSC($D$) (BSC($p$)) channel.
In \cite{korada2010polar}, $\delta_N(D) = 1 - \delta_N^2$ and $\delta_N(p) = \delta_N$.
Setting $\delta_N = \delta/N$, yields error probability at most $\delta$ and, by \cite[Lemmas 5, 6 and 7]{korada2010polar}, average distortion at most $D+\sqrt{2}\delta$. Thus, to meet a required distortion constraint, $D$, we design $F_s$ using a BSC($D'$) channel with $D'=D-\sqrt{2}\delta$. A similar statement can also be made regarding the individual codeword power constraint \cite[Theorem 2]{bur_str_full_version}.
In practice we set the thresholds $\delta_N(p)$ ($\delta_N(D)$, respectively)
such that the performance of a polar code under SCL decoding (encoding) with the set $F_c$ ($F_s$) yields the required error rate (distortion) performance.

\section{Analysis of $|F_c \cap F_s^c|$}
\label{sec:nestedness_analysis}
As was explained above, for low delay communications with polar codes using our SCL coding variant of the method in \cite{korada2010polar}, $|F_c \cap F_s^c|$ needs to be small (ideally zero).

Consider the definition of $F_s$ and $F_c$ in \eqref{eq:Fs}-\eqref{eq:Fc} and suppose that $\delta_N(D) = 1-\delta_N^2$ and $\delta_N(p) = \delta_N$ as in \cite{korada2010polar}. Then
\bre
F_c \cap F_s^c = \left\{ i \: : \: Z^{(i)}_N(p)  \ge \delta_N \:\:{\rm and}\:\: Z^{(i)}_N(D)  < 1 - \delta_N^2 \right\}
\label{eq:FcFsdef}
\ere
Define
\bre
\tilde{F}_c = \left\{ i \: : \: Z^{(i)}_N(p)  \ge 1-\delta_N^2 \right\}
\ere
The analysis in \cite{korada2010polar} asserts the following
\bre
|F_c \cap F_s^c| \le \left| F_c \setminus \tilde{F}_c \right| = o(N)
\label{eq:FcFs_ineq}
\ere
where the inequality is due to the degradedness of the sub-channel $W_N^{(i)}(D)$, corresponding to a BSC($D$), with respect to $W_N^{(i)}(p)$, corresponding to a BSC($p$) \cite{korada2010polar}. The equality is due to the polarization of the BSC($p$) channel.
A more refined argument, using scaling results of polar codes \cite{hassani2014finite, goldin2014improved} shows
\bre
|F_c \cap F_s^c| \le \left| F_c \setminus \tilde{F}_c \right| = O(N^{1-\alpha})
\label{eq:FcFsON}
\ere
for $\alpha=(1+1/0.2127)=0.175$ (this can be verified using the proofs of Theorem 1 and Theorem 2 in \cite{goldin2014improved}).

However, we observed that for small to moderate values of $N$, in the above bound of \cite{korada2010polar} (now formulated in terms of general threshold values, $\delta_N(D)$ and $\delta_N(p)$),
\bre
|F_c \cap F_s^c| \le \left| \hat{F}_c \right| \defined \left| \left\{ i \: : \: 
\delta_N(p) \le Z^{(i)}_N(p) < \delta_N(D) \right\} \right|
\ere
$|\hat{F}_c|$ is quite large for actual practical thresholds, $\delta_N(D)$ and $\delta_N(p)$, in the definitions of $F_s$ and $F_c$.
Fortunately, we have observed empirically that even though $|\hat{F}_c|$ tends to be relatively large, $|F_c \cap F_s^c|$ tends to be much smaller, and it vanishes for sufficiently large $D$ or sufficiently small $p$.
For example, consider the case where $L_c=L_s=8$ and frozen set thresholds designed for block error rate below $0.001$ for channel crossover $p$, and average distortion below $D$. Then for $N=1024$ and $p\in \{0.11,0.21,0.31\}$ we have $F_c \cap F_s^c = \emptyset$ for $D-p\ge 0.1$. For larger values of $N$, $|F_c\cap F_s^c|$ vanishes even starting from smaller values of $D-p$.
On the other hand, $|\hat{F}_c|$ is much larger, e.g. for $p=0.11$ and $D=0.25,0.45$ we have $|\hat{F}_c / N| = 0.175, 0.207$ for blocklength $N=1024$,  and $|\hat{F}_c / N| = 0.157, 0.18$ for $N=2048$.

We now study the behavior of the set $F_c \cap F_s^c$, which represents the deviation from perfect code nestedness, and prove that for $p$ sufficiently small or $D$ sufficiently large, $|F_c \cap F_s^c| = O(N^\xi)$ where $\xi>0$ can be chosen arbitrarily small, thus improving \eqref{eq:FcFsON} significantly for the case of sufficiently small $p$ or sufficiently large $D$. In fact, we prove this result for any pair of binary memoryless symmetric (BMS) channels, $W(p)$ and $W(D)$, with Bhattacharyya parameters $Z(p)$ and $Z(D)$, without requiring degradedness of $W(D)$ with respect to $W(p)$, which is important for the generalization of the results for side information channels beyond binary DP.

Consider the random processes $Z_n(p)$ and $Z_n(D)$, $n=0,1,\ldots,\log N$. They both follow the same sequence of Arikan's channel transformations, defined by \cite[Eq. (22)]{arikan2009channel} if $B_n=0$ and by \cite[Eq. (23)]{arikan2009channel} if $B_n=1$, where $\left<B_1,B_2,\ldots,B_{\log N}\right>$ defines the index of some polar sub-channel. Initially $Z_0(p) = Z(p)$ and $Z_0(D) = Z(D)$. Denote
\bre
	\epsilon_{1,n} \defined Z_n(p), \qquad
	\epsilon_{2,n} \defined 1-Z_n(D)
	\label{eq:eps1_def}\\
\ere
In particular, $\epsilon_{1,0} = Z(p)$ and $\epsilon_{2,0} = 1-Z(D)$.
Now, if $B_{n+1}=1$ then
\begin{align}
	\lefteqn{(\epsilon_{1,n+1},\epsilon_{2,n+1}) = \left( \epsilon_{1,n}^2,1-(Z_n(D))^2 \right)} \qquad\\
	&= \left(\epsilon_{1,n}^2,1-(1-\epsilon_{2,n})^2\right) = \left( \epsilon_{1,n}^2,2 \epsilon_{2,n} - \epsilon_{2,n}^2 \right)
	\label{eq:op_plus}
\end{align}
If $B_{n+1}=0$ then
\begin{align}
(\epsilon_{1,n+1},\epsilon_{2,n+1}) &\le
\left( 2\epsilon_{1,n} - \epsilon_{1,n}^2,\psi(\epsilon_{2,n}) \right)
\label{eq:op_minus}\\
\psi(x) &\defined 1-(1-x)\sqrt{1+2x-x^2}
\end{align}
where the inequality actually denotes two inequalities, one for each term. These inequalities follow from the following well known relations, e.g. \cite{arikan2009channel}, \cite[Eq. (13)]{hassani2014finite}, for $B_{n+1}=0$,
\begin{align}
Z_{n+1}(p) &\le 2Z_n(p) - Z_n(p)^2
\label{eq:BEC_extreme}\\
Z_{n+1}(D) &\ge Z_n(D) \sqrt{2-(Z_n(D))^2}
\label{eq:BSC_extreme}
\end{align}

\begin{lemma} \label{lem:epsilon_tilda}
	Consider the process $(\tilde{\epsilon}_{1,n},\tilde{\epsilon}_{2,n})$ defined by $\tilde{\epsilon}_{1,0} = Z(p)$, $\tilde{\epsilon}_{2,0}=1-Z(D)$, and, for $n=1,2,\ldots,\log N$,
	\bre
	(\tilde{\epsilon}_{1,n+1},\tilde{\epsilon}_{2,n+1}) = 
	\left\{
	\begin{array}{ll}
		\left( \tilde{\epsilon}_{1,n}^2,2 \tilde{\epsilon}_{2,n} - \tilde{\epsilon}_{2,n}^2 \right) & \hbox{if $B_{n+1}=1$} \\
		\left( 2\tilde{\epsilon}_{1,n} - \tilde{\epsilon}_{1,n}^2,\psi(\tilde{\epsilon}_{2,n})
		\right)  & \hbox{if $B_{n+1}=0$}
	\end{array}
	\right.
	\label{eq:next_eps}
	\ere
	For $n=\log N$ we have $N$ possible realizations of the process corresponding to all possible sub-channels. The number of realizations for which both $\tilde{\epsilon}_{1,\log N} \ge \delta_N(p)$ and $\tilde{\epsilon}_{2,\log N} > 1-\delta_N(D)$ is an upper bound on $|F_c \cap F_s^c|$.
\end{lemma}
The proof follows from \eqref{eq:op_plus}-\eqref{eq:op_minus} and the fact that all the functions that appear in \eqref{eq:next_eps}, including $\psi()$, are monotonically increasing for $\tilde{\epsilon}_{i,n} \in (0,1)$, $i=1,2$.

We note that the bound provided by Lemma \ref{lem:epsilon_tilda} on $|F_c \cap F_s^c|$ is monotonically increasing in $Z(p)$ and monotonically decreasing in $Z(D)$.
We used Lemma \ref{lem:epsilon_tilda} to compute bounds on $|F_c \cap F_s^c|$ for some $p$ and $D$ values, and compare with the actual value of $|F_c \cap F_s^c|$. $\delta_N(p)$ and $\delta_N(D)$ were set to obtain a block error probability $0.001$, and average distortion $D$ with $L_c=8$ and $L_s=8$. 
As an example, for $N=1024$ and $p=0.11$ ($p=0.21$, respectively), $|F_c \cap F_s^c|$ vanishes for $D-p \ge 0.1$ ($D-p \ge 0.1$) while the bound requires $D-p \ge 0.16$ ($D-p \ge 0.14$). For $N=2048$ and $p=0.11$ (same results for $p=0.21$), $|F_c \cap F_s^c|$ vanishes for $D-p \ge 0.08$ while the bound requires $D-p \ge 0.14$.

We proceed the analysis by defining $\hepsilon_{1,n}$, $\hepsilon_{2,n}$ by
\bre
	\hepsilon_{1,n} \defined Z_n(p), \qquad
	\hepsilon_{2,n} \defined 1- \left(Z_n(D)\right)^2 \label{eq:heps1_def}
\ere
such that $\hepsilon_{1,n} = \epsilon_{1,n}$ (see \eqref{eq:eps1_def}).
Similarly to \eqref{eq:op_plus} we have that if $B_{n+1}=1$ then
\begin{align}
	\lefteqn{
		\left( \hepsilon_{1,n+1}, \hepsilon_{2,n+1} \right) = \left( \hepsilon_{1,n}^2, 1-(Z_n(D))^4 \right) =
	}\qquad\qquad\\
	&\left( \hepsilon_{1,n}^2, 1-(1-\hepsilon_{2,n})^2 \right) = \left( \hepsilon_{1,n}^2, 2\hepsilon_{2,n} - \hepsilon_{2,n}^2 \right)
\end{align}

Similarly to \eqref{eq:op_minus}, if $B_{n+1}=0$ then
\bre
\left( \hepsilon_{1,n+1}, \hepsilon_{2,n+1} \right)
\le 
\left( 2\hepsilon_{1,n} - \hepsilon_{1,n}^2,\hepsilon_{2,n}^2\right)
\ere
The inequality for the left terms is due to \eqref{eq:BEC_extreme} and the inequality for the right terms is due to \eqref{eq:BSC_extreme} that can be rewritten as
\bre
\left(Z_{n+1}(D)\right)^2 \ge \left( Z_n(D) \right)^2 \left( 2 - \left(Z_n(D)\right)^2 \right)
\ere
Thus we have
\bre
(\hepsilon_{1,n+1},\hepsilon_{2,n+1}) 
\left\{
\begin{array}{ll}
	=   \left( \hepsilon_{1,n}^2,2 \hepsilon_{2,n} - \hepsilon_{2,n}^2 \right) & \hbox{if $B_{n+1}=1$} \\
	\le \left( 2\hepsilon_{1,n} - \hepsilon_{1,n}^2,\hepsilon_{2,n}^2  \right) & \hbox{if $B_{n+1}=0$}
\end{array}
\right.
\label{eq:BMS_eps_transform}
\ere

We now claim the following key lemma.
\begin{lemma}
	\bre
	\hepsilon_{1,n} \hepsilon_{2,n} < \gamma \quad \forall n
	\ere
	where $\gamma = \gamma(Z(p),Z(D))$ becomes arbitrarily small for $Z(p)$ sufficiently small or $Z(D)$ sufficiently large.
	\label{lem:BMS}
\end{lemma}

\begin{proof}
	For notational convenience, denote by $\hepsilon_1\defined\hepsilon_{1,n}$, $\hepsilon_2\defined\hepsilon_{2,n}$, and
	\bre
	R_n \defined \log \frac{\hepsilon_1}{1-\hepsilon_1} + \log \frac{\hepsilon_2}{1-\hepsilon_2}
	\ere
	Hence,
	\bre
	R_{n+1} 
	\left\{
	\begin{array}{ll}
		= \log \frac{\hepsilon_1^2}{1-\hepsilon_1^2} + \log \frac{2\hepsilon_2-\hepsilon_2^2}{(1-\hepsilon_2)^2}
		& \hbox{if $B_{n+1}=1$} \\
		\le \log \frac{2\hepsilon_1-\hepsilon_1^2}{(1-\hepsilon_1)^2} + \log \frac{\hepsilon_2^2}{1-\hepsilon_2^2}
		& \hbox{if $B_{n+1}=0$}
	\end{array}
	\right.
	\label{eq:Rn1_}
	\ere
	We will first show that for all $n$, if $R_n\le A$, where $A<0$, then $R_{n+1} \le A$. For that, it is sufficient to consider the first case in \eqref{eq:Rn1_} ($B_{n+1}=1$), since the same proof holds for the other case, $B_{n+1}=0$. Now, the first case in \eqref{eq:Rn1_} can be written as
	\bre
	R_{n+1} = R_n + \log \frac{\hepsilon_1 (2-\hepsilon_2)}{(1+\hepsilon_1)(1-\hepsilon_2)}
	\label{eq:Rn11_A}
	\ere
	Since $R_n<A<0$, we have
	\bre
	\frac{\hepsilon_1}{1-\hepsilon_1} \cdot \frac{\hepsilon_2}{1-\hepsilon_2} < 1
	\ere
	Hence, $\hepsilon_1 + \hepsilon_2 < 1$. Therefore,
	\bre
	\frac{2-\hepsilon_2}{1-\hepsilon_2} = 1 + \frac{1}{1-\hepsilon_2} < 1 + \frac{1}{\hepsilon_1} = \frac{\hepsilon_1+1}{\hepsilon_1}
	\ere
	Using this inequality in \eqref{eq:Rn11_A} yields $R_{n+1}<R_n\le A$ as claimed. We conclude that if $R_0=A<0$ then $R_n\le R_0$ for all $n$, so that
	\begin{align}
		\lefteqn{
			\log \hepsilon_{1,n} + \log \hepsilon_{2,n} < R_n \le R_0 =}\qquad\\ 
		&\log \frac{Z(p)}{1-Z(p)} + \log \frac{1-Z^2(D)}{Z^2(D)} \defined \log \gamma(Z(p),Z(D))
	\end{align}
	where the first inequality follows from the fact that $1-\hepsilon_{i,n}<1$, for $i=1,2$.
	Note that $\gamma(Z(p),Z(D))$ can be made arbitrarily small by choosing $Z(p)$ sufficiently small or $Z(D)$ sufficiently large.
\end{proof}
We can now state and prove our main result for the nested codes property.
\begin{theorem}
	Consider the case where $W(p)$ and $W(D)$ are BMS channels with Bhattacharyya parameters $Z(p)$ and $Z(D)$. Given $0<Z(p)<Z(D)<1$, suppose either a small enough $Z(p)$ or a large enough $Z(D)$. Then, $|F_c \cap F_s^c| = O(N^\xi)$ where $\xi>0$ can be set arbitrarily small.
	\label{thm:main_result_BMS}
\end{theorem}
\begin{proof}
	By \eqref{eq:BMS_eps_transform} we have
	\bre
	\hat{\epsilon}_{1,n+1} \hat{\epsilon}_{2,n+1} \le
	\left\{
	\begin{array}{ll}
		\hat{\epsilon}_{1,n}^2 \hat{\epsilon}_{2,n} (2 - \hat{\epsilon}_{2,n}) & \hbox{if $B_{n+1}=1$} \\
		\hat{\epsilon}_{1,n}(2 - \hat{\epsilon}_{1,n})\hat{\epsilon}_{2,n}^2  & \hbox{if $B_{n+1}=0$}
	\end{array}
	\right.
	\ere
	Hence,
	\bre
	\rho_n \defined
	\frac{\hat{\epsilon}_{1,n+1} \hat{\epsilon}_{2,n+1}}{\hat{\epsilon}_{1,n}\hat{\epsilon}_{2,n}} \le 
	\left\{
	\begin{array}{ll}
		\hat{\epsilon}_{1,n}(2-\hat{\epsilon}_{2,n})   & \hbox{if $B_{n+1}=1$} \\
		\hat{\epsilon}_{2,n}(2-\hat{\epsilon}_{1,n})   & \hbox{if $B_{n+1}=0$}
	\end{array}
	\right.
	\ere
	Now, since by Lemma \ref{lem:BMS} either $\hat{\epsilon}_{1,n} < \sqrt{\gamma}$ or $\hat{\epsilon}_{2,n} < \sqrt{\gamma}$, we conclude that w.p. $1/2$, $\rho_n \le 2$ and w.p. $1/2$, $\rho_n \le 2\sqrt{\gamma}$. That is,
	\bre
	\frac{\hat{\epsilon}_{1,n+1} \hat{\epsilon}_{2,n+1}}{\hat{\epsilon}_{1,n}\hat{\epsilon}_{2,n}} \le
	\left\{
	\begin{array}{ll}
		2               & \hbox{if $\tilde{B}_{n+1}=0$} \\
		2\sqrt{\gamma}  & \hbox{if $\tilde{B}_{n+1}=1$}
	\end{array}
	\right.
	\ere
	where, similarly to $\{B_i\}$, the random variables $\{\tilde{B}_i\}$ are independent, binary, uniformly distributed (i.e., $\tilde{B}_i = (0,1)$ w.p. $(1/2,1/2)$).
	
	Following \cite[Section IV.B]{arikan2009channel}, define, for $\eta \in (0,1/2)$, the event
	\bre
	\cU_n(\eta) \defined \left\{ \sum_{i=1}^{n} \tilde{B}_i > (\frac{1}{2} - \eta) n  \right\}
	\ere
	Using the same argument as in \cite[Section IV.B]{arikan2009channel} we know that if the event $\cU_n(\eta)$ holds then
	\bre
	\hat{\epsilon}_{1,n}\hat{\epsilon}_{2,n} \le
	\zeta \cdot \left[ 2^{0.5+\eta} \zeta^{0.5-\eta} \right]^n
	\label{eq:epsilon_product}
	\ere
	for $\zeta=2\sqrt{\gamma}$. It is also known \cite[Section IV.B]{arikan2009channel}, by Chernoff's bound, that
	\bre
	P(\cU_n(\eta)) \ge 1 - 2^{-n[1-h_2(0.5-\eta)]}
	\ere
	Now, \eqref{eq:epsilon_product} implies that
	\bre
	\min(\hat{\epsilon}_{1,n},\hat{\epsilon}_{2,n}) \le
	\sqrt{\zeta} \cdot \left[ 2^{0.5+\eta} \zeta^{0.5-\eta} \right]^{n/2}
	\ere
	Setting $n=\log N$ we obtain
	\bre
	P\left(
	\min\left( Z_{\log N}(p), 1-Z_{\log N}^2(D) \right) \le N^{-a} 
	\right)
	\ge 1 - N^{-1+\xi}
	\ere
	where $\xi=h_2(0.5-\eta)$. Furthermore, $\xi>0$ can be made arbitrarily small by setting $\eta \rightarrow 0.5^-$.
	Recall that if either $Z(p)$ is small enough or $Z(D)$ large enough, then $\gamma$ can be made as small as desired. Hence $a$ can be set as large as desired (any $a>2$ is sufficient to prove the theorem).
	Recalling the connection between the random process $Z_n(p)$ ($Z_n(D)$, respectively) for $n=\log N$ and the values of the sub-channels, $Z_N^{(i)}(p)$ ($Z_N^{(i)}(D)$) \cite{arikan2009channel}, we obtain
	\bre
	\left| \left\{ i \: : \:
	Z_N^{(i)}(p) > N^{-a} 
	\:{\rm and }\:\:
	Z_N^{(i)}(D) < \sqrt{1-N^{-a}}
	\right\} \right|
	< N^{\xi}
	\ere
	Combining this with \eqref{eq:FcFsdef} concludes the proof (since we can take $a>2$).
\end{proof}

\section {Simulation Results}\label{sec:simulation_review_fin}
We now present results for our polar SCL scheme to the binary DP problem.
All the results presented here were achieved without the need to use retransmission, i.e., $F_s \cap F_c^s = \emptyset$ in all the reported cases.

Figure \ref{fig:DRGP1} presents our results for $p=0.11$ and $D\in{0.21,0.31,0.41}$. We used $L_c=8$ lists and CRC of size 8 in the decoder, and $L_s \in \{1,50\}$ in the encoder. For each experiment the figure shows the maximum polar SCL rate under the constraint of error rate below $\epsilon_p=0.001$ and average distortion below $D$.
It can be seen how increasing the number of lists in the source encoder increases the achievable rate.
\begin{figure}[htbp]
	\centering
	\includegraphics[width=1.1\columnwidth]{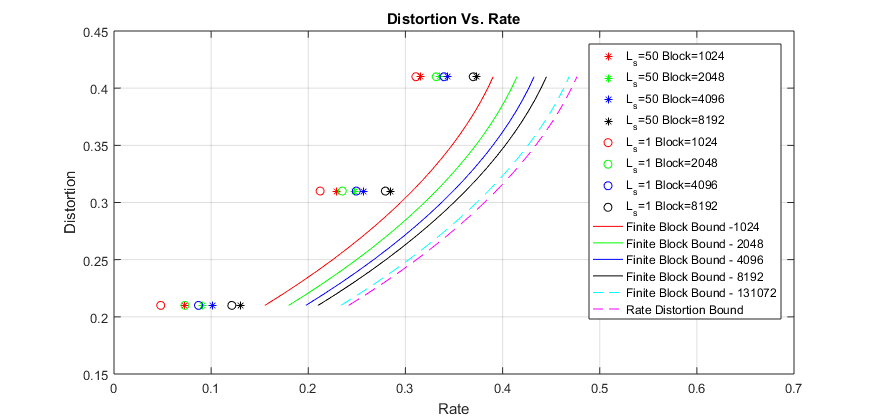} 
	\vspace{-10pt}
	\caption{Distortion vs. rate for the binary DP problem with $p=0.11$ and $D\in \left\{0.21,0.31,0.41\right\}$. Simulation used $L_c=8$ and various blocklengths and $L_s$ values. Block error rate is at most $10^{-3}$ for all experiments.}
	\label{fig:DRGP1}
\end{figure}
The figure also shows the approximated maximum achievable rate of any nested coding scheme for binary DP. It was obtained using the approximated maximum achievable channel coding rate, and minimum achievable lossy compression rate in a finite blocklength regime, in \cite[Theorem 52]{polyanskiy2010channel} and \cite[Eqs. (1), (11) and (93)]{kostina2012fixed} respectively. Since the code is nested, its approximated binary DP maximum achievable rate is obtained by subtracting these two approximations from \cite{polyanskiy2010channel} and \cite{kostina2012fixed}.
Let $N$ be the blocklength. Denote by $\epsilon_p$ the block error rate, and by $\epsilon_D$ the distortion constraint violation rate. Then
\bre
R_{GP}\left(N,D,p,\epsilon_p,\epsilon_D\right) = \overline{R}_{D} - \overline{R}_{p}
\ere
\begin{align} 
	\overline{R}_{D}     &\approx 1-h_2(D)+\sqrt{\frac{V(D)}{N}}Q^{-1}\left(\epsilon_D\right)
	\\
	V(D) &= D \log^2 D + (1-D) \log^2 (1-D) - h_2^2(D)
	\\
	\overline{R}_{p}     &\approx 1-h_2(p)+\frac{\log N}{2N}
	\\
	&- \sqrt\frac{p(1-p)}{N} \log \left(\frac{1-p}{p}\right)Q^{-1}\left(\epsilon_p\right)
\end{align}
where $Q^{-1}()$ is the inverse of the standard Gaussian complementary cumulative distribution function.
The approximated bounds in Fig. \ref{fig:DRGP1} were obtained by setting $\epsilon_p = 0.001$ and $\epsilon_D=0.5$, which means that we are taking the standard rate distortion bound for the lossy source coding part (since in this experiment we only set an average distortion constraint, as in \cite{korada2010polar}). Repeating the same experiment with $\epsilon_D=0.01$ yields an even smaller gap between the achievable rates using the SCL polar coding scheme and the bounds.

We have also compared our polar SCL coding scheme to the superposition coding scheme in \cite{bennatan2006scs} for long blocklength codes, using LDPC codes for channel coding and convolutional codes for source coding. In \cite{bennatan2006scs} the blocklength was $N$=100,000. In our experiments we used both $N=2^{17}=\mbox{131,072}$ and $N=2^{16}=\mbox{65,536}$. The results, shown in Table \ref{tb:scl_ldpc}, show comparable results for long blocklength codes.
The SCL results were tested 100 times as in \cite{bennatan2006scs}.
We note that the method in \cite{bennatan2006scs} required considerable computational resources (150--200 belief propagation iterations and 10--15 BCJR iterations with 1024 states in the decoding trellis). Results for shorter blocklengths are not reported in \cite{bennatan2006scs}.
\begin{table}[H]
	\caption{LDPC v.s. SCL binary DP performance, $p=0.1$, $D=0.3$}
	\label{tb:scl_ldpc}
	\begin{tabular}{|l|l|l|l|l|l|}
		\hline
		Description & Block size                & DP Rate & Bit Error Rate   \\ \hline
		LDPC scheme & 100K                          & 0.36             & $1.28\times{10^{-5}}$        \\ \hline
		SCL, $L_c=4$, $L_s=1$ & 130K & 0.353      & $5.9\times{10^{-6}}$         \\ \hline
		SCL, $L_c=8$, $L_s=8$ & 130K & 0.357      & $5.6\times{10^{-6}}$         \\ \hline
		SCL, $L_c=16$, $L_s=50$ & 130K & 0.362      & $5.0\times{10^{-6}}$         \\ \hline
		SCL, $L_c=16$, $L_s=50$ & 65K & 0.356      & $1.2\times{10^{-5}}$         \\ \hline
		Capacity  &                                &   0.42               &                \\ \hline
	\end{tabular}
\end{table}
\section*{Acknowledgment}
This research was supported by the Israel Science Foundation (grant no. 1868/18).
\bibliographystyle{IEEEtran}
\bibliography{bibliography}

\begin{thebibliography}{10}
\providecommand{\url}[1]{#1}
\csname url@samestyle\endcsname
\providecommand{\newblock}{\relax}
\providecommand{\bibinfo}[2]{#2}
\providecommand{\BIBentrySTDinterwordspacing}{\spaceskip=0pt\relax}
\providecommand{\BIBentryALTinterwordstretchfactor}{4}
\providecommand{\BIBentryALTinterwordspacing}{\spaceskip=\fontdimen2\font plus
\BIBentryALTinterwordstretchfactor\fontdimen3\font minus
  \fontdimen4\font\relax}
\providecommand{\BIBforeignlanguage}[2]{{%
\expandafter\ifx\csname l@#1\endcsname\relax
\typeout{** WARNING: IEEEtran.bst: No hyphenation pattern has been}%
\typeout{** loaded for the language `#1'. Using the pattern for}%
\typeout{** the default language instead.}%
\else
\language=\csname l@#1\endcsname
\fi
#2}}
\providecommand{\BIBdecl}{\relax}
\BIBdecl

\bibitem{korada2010polar}
S.~B. Korada and R.~L. Urbanke, ``{Polar codes are optimal for lossy source
  coding},'' \emph{IEEE Transactions on Information Theory}, vol.~56, no.~4,
  pp. 1751--1768, 2010.

\bibitem{arikan2009channel}
E.~Arikan, ``{Channel polarization: A method for constructing
  capacity-achieving codes for symmetric binary-input memoryless channels},''
  \emph{IEEE Transactions on Information Theory}, vol.~55, no.~7, pp.
  3051--3073, 2009.

\bibitem{bur_str_full_version}
D.~Burshtein and A.~Strugatski, ``{Polar write once memory codes},'' \emph{IEEE
  Transactions on Information Theory}, vol.~59, no.~8, pp. 5088--5101, August
  2013.

\bibitem{hassani2014universal}
S.~H. Hassani and R.~Urbanke, ``Universal polar codes,'' in \emph{IEEE
  International Symposium on Information Theory (ISIT)}, 2014, pp. 1451--1455.

\bibitem{mondelli2015achieving}
M.~Mondelli, S.~H. Hassani, I.~Sason, and R.~L. Urbanke, ``{Achieving Marton's
  region for broadcast channels using polar codes},'' \emph{IEEE Transactions
  on Information Theory}, vol.~61, no.~2, pp. 783--800, February 2015.

\bibitem{tal2015list}
I.~Tal and A.~Vardy, ``List decoding of polar codes,'' \emph{IEEE Transactions
  on Information Theory}, vol.~61, no.~5, pp. 2213--2226, 2015.

\bibitem{hassani2014finite}
S.~H. Hassani, K.~Alishahi, and R.~Urbanke, ``Finite-length scaling for polar
  codes,'' \emph{IEEE Transactions on Information Theory}, vol.~60, no.~10, pp.
  5875--5898, 2014.

\bibitem{goldin2014improved}
D.~Goldin and D.~Burshtein, ``{Improved bounds on the finite length scaling of
  polar codes},'' \emph{IEEE Transactions on Information Theory}, vol.~60,
  no.~11, pp. 6966--6978, November 2014.

\bibitem{polyanskiy2010channel}
Y.~Polyanskiy, H.~V. Poor, and S.~Verd{\'u}, ``Channel coding rate in the
  finite blocklength regime,'' \emph{IEEE Transactions on Information Theory},
  vol.~56, no.~5, pp. 2307--2359, 2010.

\bibitem{kostina2012fixed}
V.~Kostina and S.~Verd{\'u}, ``Fixed-length lossy compression in the finite
  blocklength regime,'' \emph{IEEE Transactions on Information Theory},
  vol.~58, no.~6, pp. 3309--3338, 2012.

\bibitem{bennatan2006scs}
A.~Bennatan, D.~Burshtein, G.~Caire, and S.~Shamai, ``{Superposition coding for
  side-information channels},'' \emph{IEEE Transactions on Information Theory},
  vol.~52, no.~5, pp. 1872--1889, May 2006.

\end{thebibliography}
\end{document}